\documentclass[letterpaper, unpublished]{quantumarticle}
\pdfoutput=1

\usepackage[english]{babel}

\usepackage{amsmath, amssymb, amsfonts}
\usepackage{graphicx}
\usepackage[colorlinks=true, allcolors=blue]{hyperref}
\usepackage{braket}
\usepackage[ruled]{algorithm2e}
\usepackage{makecell}

\usepackage{amsthm, thmtools}
\usepackage{thm-restate}
\newtheorem{theorem}{Theorem}
\newtheorem{definition}{Definition}
\newtheorem{lemma}{Lemma}

\usepackage[backend=biber, style=alphabetic]{biblatex}
\usepackage[dvipsnames]{xcolor}

\newif\ifnotes
\notesfalse

\newcommand\qft{{\mathsf{QFT}}}
\newcommand\oqft{{\widetilde{\qft}}}
\newcommand\U{{\mathcal{U}}}
\newcommand\oU{{\widetilde{\U}}}
\newcommand\C{{\mathcal{C}}}
\newcommand\E{{\mathcal{E}}}
\renewcommand\H{{\mathcal{H}}}
\newcommand\Hbad{{\H_\mathrm{bad}}}
\newcommand\imax{{i_\mathrm{max}}}

\newcommand\Xp{{\ket{\widetilde{X_i}}}}

\newcommand{\add}{\mathsf{Add}}

\title{A log-depth in-place quantum Fourier transform that rarely needs ancillas}
\author{Gregory D. Kahanamoku-Meyer}
\thanks{gkm@mit.edu}
\affiliation{Massachusetts Institute of Technology, Cambridge, MA 02139}
\author{John Blue}
\affiliation{Massachusetts Institute of Technology, Cambridge, MA 02139}
\author{Thiago Bergamaschi}
\affiliation{University of California, Berkeley, CA 96720}
\author{Craig Gidney}
\affiliation{Google, Inc., Santa Barbara, CA 93111}
\author{Isaac L. Chuang}
\affiliation{Massachusetts Institute of Technology, Cambridge, MA 02139}
\date{}

\begin{document}

\maketitle

\onecolumn

\begin{abstract}
When designing quantum circuits for a given unitary, it can be much cheaper to achieve a good approximation on most inputs than on all inputs.
In this work we formalize this idea, and propose that such ``optimistic quantum circuits'' are often sufficient in the context of larger quantum algorithms.
For the rare algorithm in which a subroutine needs to be a good approximation on all inputs, we provide a reduction which transforms optimistic circuits into general ones.
Applying these ideas, we build an optimistic circuit for the in-place quantum Fourier transform (QFT).
Our circuit has depth $O(\log (n / \epsilon))$ for tunable error parameter $\epsilon$, uses $n$ total qubits, i.e. no ancillas, is local for input qubits arranged in 1D, and is measurement-free.
The circuit's error is bounded by $\epsilon$ on all input states except an $O(\epsilon)$-sized fraction of the Hilbert space.
The circuit is also rather simple and thus may be practically useful.
Combined with recent QFT-based fast arithmetic constructions (\href{https://arxiv.org/abs/2403.18006}{arxiv.org/2403.18006}), the optimistic QFT yields factoring circuits of nearly linear depth using only $2n + O(n/\log n)$ total qubits.
Additionally, we apply our reduction technique to yield an approximate QFT with well-controlled error on all inputs; it is the first to achieve the asymptotically optimal depth of $O(\log (n/\epsilon))$ with a sublinear number of ancilla qubits.
The reduction uses long-range gates but no measurements.
\end{abstract}

\vspace{1em}

\twocolumn

\section{Introduction}\label{introduction}

In classical computing, there exist an abundance of algorithms which perform better in the average case than in the worst case.
Occasional ``bad'' (inefficient) inputs do not substantially detract from the overall performance of an algorithm if the cost can be amortized across many function calls.
Moreover, there exist strategies for avoiding an
algorithm's worst-case behavior entirely: it may be possible to detect bad inputs early and switch to an algorithm which is more performant on them, or to solve the worst case by building off of the solution of a related average-case input---a technique known as a worst-case to average-case reduction
\cite{impagliazzo_personal_1995,bogdanov_average-case_2006,goldreich_notes_2011,asadi_quantum_2024}.

However, the application of these techniques to the design of quantum circuits requires overcoming considerable obstacles.
One of the most fundamental is that we cannot conditionally execute different algorithms depending on the input---we may receive a superposition of ``good'' and ``bad'' inputs, and thus we must follow every code path, every time!
On the other hand, quantum mechanics provides us with unique tools to overcome such obstacles.
By linearity a sufficiently good \emph{approximation} of a desired unitary can be used in place of the exact unitary, without meaningfully affecting the measurement results at the end of a quantum computation.
Typically, this fact has been used to find approximate circuits whose \emph{maximum} error over arbitrary input states is small.
Unfortunately, such constructions cannot leverage the fact that much of the complexity in a quantum circuit may be due to a small ``corner case'' subspace of possible inputs, that is more difficult to approximate.

In this work, we explore the design of quantum circuits which give a good approximation of the desired unitary only on \emph{most}, but not all, of the input Hilbert space.
We call these ``optimistic quantum circuits.''%
\footnote{We do not use the term ``average case'' for our circuits because the average case is defined with respect to a distribution over inputs, which depends on the context in which an algorithm is being used.
We define optimistic quantum circuits in a way that is context- and even basis-independent (see Definition~\ref{def:opt-circuits-basis-independent}).}
We propose to use such circuits in two different ways: either 1) as drop-in replacements for the desired unitary, if it can be shown (or perhaps heuristically assumed) that pathological inputs are unlikely to occur, or 2) as the core of a worst-to-average case reduction which yields a provably good approximation even on worst-case inputs.
We give an explicit way of constructing such a reduction.

Applying these ideas, we construct an optimistic circuit
for the quantum Fourier transform (QFT), which has a set of features that are provably impossible to achieve simultaneously for a ``full'' approximate QFT.
For error parameter $\epsilon$ and input size $n$ qubits, the circuit achieves depth $O(\log (n/\epsilon))$ with no ancilla qubits, long-range gates, nor measurements.
This can be compared to previous logarithmic-depth QFT constructions, which must relax at least one of these constraints.
A first construction required $O(n\log n)$ ancilla qubits~\cite{cleve_fast_2000}, which was improved to $O(n)$ ancillas by Hales~\cite{hales_quantum_2002}.
Both constructions require long-range gates; very recently the connectivity was improved to nearest-neighbor (in 1D) through the use of measurement and feedforward while maintaining $O(n)$ ancillas~\cite{baumer_approximate_2025}.
The worst-to-average-case reduction we describe in this work for the QFT avoids measurements, trading them for long-range gates and a sublinear number of ancilla qubits.
(For an extensive comparison to previous works, see Appendix~\ref{app:previous-qft}.)

The key idea behind our optimistic QFT circuit is that leveraging quantum phase estimation on blocks of qubits allows us to approximately map the output state on that block back to the input, which we may then use to control phase rotations on a neighboring block (Fig~\ref{fig:oqft-detail}a).
This phase estimation fails on basis states for which the phase estimation ``wraps around'' modulo $2^m$, where $m$ is the block size (Fig~\ref{fig:oqft-detail}b); fortunately, the error from this wraparound is small on the vast majority of the Hilbert space and thus the construction yields an optimistic QFT.
We show that our optimistic QFT can be used with recent QFT based arithmetic circuits to construct an optimistic multiplier~\cite{kahanamoku-meyer_fast_2024}, and that this multiplier is sufficient to implement Shor's algorithm for factoring with high probability of success.
We also instantiate the reduction discussed earlier, and show that a randomized reduction yields an approximate QFT having low error on arbitrary inputs in depth $O(\log(n/\epsilon))$ using $n + O(n/\log(n/\epsilon))$ total qubits.
A straightforward purification of the randomized reduction yields a derandomized approximate QFT with the same depth and $3n+O(n/\log(n/\epsilon))$ total qubits.

We begin below by formally defining optimistic quantum circuits and presenting some abstract results about their behavior (Section~\ref{sec-opt-circuits}), then specialize to present an ancilla-free logarithmic-depth optimistic circuit construction for the QFT (Section~\ref{sec:opt-qft}).  This construction can be used to build an optimistic multiplier that reduces resources required for factoring (Section~\ref{sec:opt-qft-factoring}).  An approximate QFT which works on arbitrary inputs is also given (Section~\ref{sec:oqft-reduction}) before we conclude with an outlook (Section~\ref{sec-discussion}).  Proofs are largely deferred to the appendices.

\section{Optimistic quantum circuits}
\label{sec-opt-circuits}

Intuitively, an optimistic quantum circuit is one in which the error is small on most basis states (for any basis).
Definitionally, it is more convenient to consider the \emph{average} error over all basis states, as that turns out to be a basis-independent quantity.
Thus, we begin by presenting such a definition and then in Lemma~\ref{lem:large-error-subspaces} show that the definition captures the intuition of small error on all but a few states in any basis.

\begin{definition}[Optimistic quantum circuits]
	\label{def:opt-circuits}
	Let $\U$ be a unitary operation on a Hilbert space $\H$.
	Consider a quantum circuit $\C$ which induces a unitary $\oU$ on $\H$.
	$\C$ is an optimistic quantum circuit for $\U$ with error bound $\epsilon$ if, for any complete set of orthonormal basis states $\{\ket{\phi_i}\}$ for $\H$,
	\begin{equation}
		\label{eq:opt-circuit-def-bound}
		\frac{1}{\dim{\H}} \sum_i ||\oU \ket{\phi_i} - \U \ket{\phi_i} ||^2 < \epsilon.
	\end{equation}
\end{definition}

This formulation is convenient because it gives us a clear method for showing that a particular construction is an optimistic circuit with error $\epsilon$: pick an orthonormal basis, and show that the average error over states in this basis is less than $\epsilon$.
It is less obvious, but true, that this definition is \emph{basis-independent}.
This can be seen by noting that the quantity in Equation~\ref{eq:opt-circuit-def-bound} is proportional to the Frobenius norm, which inherits its basis-independence from the trace:
\begin{definition}[Equivalent to Definition~\ref{def:opt-circuits}]
	\label{def:opt-circuits-basis-independent}
	For $\U$, $\oU$, $\H$, and $\C$ as in Definition~\ref{def:opt-circuits}, $\C$ is an optimistic quantum circuit for $\U$ with error bound $\epsilon$ if
	\begin{equation}
		\frac{||\oU - \U||_F^2}{\dim{\H}} \leq \epsilon
	\end{equation}
	where $\|\E\|_F^2 = \mathrm{Tr}\left[\E^\dagger \E\right]$ is the Frobenius norm squared.
\end{definition}

Two features of Definition~\ref{def:opt-circuits} are crucial for this basis-independence to hold: first, that the per-state error is measured not via fidelity, but instead via length of the error vector, which captures errors in phase; and second, that we are concerned with the \emph{average} error over basis states, rather than the maximum.

While by design the maximum error over basis states is essentially unbounded, our original intuition for optimistic quantum circuits says that the size of the subspace in which these large errors occur must be small.
Indeed, for an optimistic quantum circuit with error parameter $\epsilon$, the number of basis states for which the error is $O(1)$ is bounded by at most $O(\epsilon) \cdot \dim{\H}$.
We formalize this idea in the following Lemma.
\begin{lemma}[Size of large-error subspaces]
	\label{lem:large-error-subspaces}
	For $\U$, $\oU$, $\H$, and $\C$ as in Definition~\ref{def:opt-circuits}, consider a subspace $\Hbad \subseteq \H$ with associated projector $\Pi_{\Hbad}$.
	Let the error on this subspace be $\mu = ||(\oU-\U)\Pi_{\Hbad}||_F^2 / \dim{\Hbad}$.
	Then, if $\C$ is an optimistic circuit for $\U$ with error parameter $\epsilon$, $\dim{\Hbad} \leq \dim \H \cdot \epsilon/\mu$.
\end{lemma}
\begin{proof}
	Observe that for any subspace $\H' \subseteq \H$, the error on $\H$ (as defined in Definition~\ref{def:opt-circuits}) is bounded from below by the error on the subspace, scaled by their relative size:
	\begin{equation*}
		\frac{||\oU - \U||_F^2}{\dim{\H}} \geq \frac{\dim{\H'}}{\dim{\H}} \left( \frac{||(\oU - \U)\Pi_{\H'}||_F^2}{\dim{\H'}} \right)
	\end{equation*}
	Substituting the definitions of $\mu$ and $\epsilon$ yields the desired result.
\end{proof}

Because in any basis the error is small on all but a few states, in practice it should almost always be possible to use optimistic quantum circuits directly in place of exact ones.
In Section~\ref{sec:opt-qft-factoring} we show explicitly that this is the case for an optimistic multiplication circuit, in the context of Shor's algorithm for factoring.
But perhaps in some cases the input states will be pathologically concentrated in the large-error subspace; to handle such cases, we next show a generic reduction by which worst-case approximate circuits can be constructed from optimistic circuits.

\subsection{Reduction: worst-case approximate circuits from optimistic circuits}
\label{subsec:reduction}

From Definition~\ref{def:opt-circuits} (or Lemma~\ref{lem:large-error-subspaces}), it follows that an optimistic circuit's error is expected to be small when applied to random input states.
In consideration of that fact, the core idea of our reduction is to apply a random unitary $V$ before the optimistic circuit, and then apply $\hat{V}^\dagger = \U V^\dagger \U^\dagger$ after it.
The action of $V$ effectively turns every input state into a random state, and $\hat{V}^\dagger$ inverts the randomization.
If efficient circuits exist for both $V$ and $\hat{V}^\dagger$, the entire operation $\hat{V}^\dagger \oU V \approx \U$ can be implemented efficiently.
We note that the idea of using randomization to avoid worst-case error behavior has appeared in an array of works across the field of quantum algorithms, e.g.~\cite{campbell_random_2019, martyn_halving_2025, dalzell_distillation-teleportation_2025, gosset_multi-qubit_2025}.

Applying truly random unitaries (i.e. drawn from the Haar measure) would be very expensive; the key result of this subsection is that requiring the unitaries to be Haar-random is unnecessarily strong.
Indeed, here we show that all we need is a unitary 1-design, which is an ensemble of unitaries $\{V_i\}$ with associated probabilities $p_i$, such that for any operator $X$, we have $\sum_i p_i V_i X V_i^\dagger = \mathbb{I} \cdot \mathrm{Tr}[X]/\dim{\H}$.
Fortunately, 1-designs have a reputation for being among the easiest unitary ensembles to construct, so we can leverage this freedom to choose one that is convenient in practice.
For example, in Section~\ref{sec:oqft-reduction} we are able to find a 1-design for the optimistic QFT where both $V$ and $\hat{V}^\dagger$ have logarithmic-depth circuits requiring few ancillas.

We formalize our randomized reduction in the following theorem; later, in Theorem~\ref{thm:purification}, we show that if the 1-design is a uniform distribution over some set of unitaries $\{V_i\}$, the reduction can be derandomized via purification.%
\footnote{For simplicity, we focus on purifying 1-designs over uniform distributions, as the purification is easily implementable.}

\begin{theorem}[Randomized reduction]
	\label{thm:randomized-reduction}
	Consider an optimistic quantum circuit $\C$ with error parameter $\epsilon$, and $\U$, $\oU$, and $\H$ defined as in Definition~\ref{def:opt-circuits}.
	Furthermore let $d$ be a unitary 1-design on $\H$.
	Then for any state $\ket{\psi} \in \H$, $\hat{V}^\dagger \oU V \ket{\psi}$ is a good approximation of $\U \ket{\psi}$ with high probability over the distribution of $V$:
	\begin{equation}\label{equation:average_error}
		\mathop{\mathbb{E}}_{V \sim d} \left[ | \hat{V}^\dagger \oU V \ket{\psi} - \U \ket{\psi} |^2 \right] \leq \epsilon
	\end{equation}
	where $\hat{V} = \U V \U^\dagger$ and $\mathop{\mathbb{E}}_{V \sim d}[\cdot]$ denotes the expectation value over unitaries $V$ drawn from $d$.
\end{theorem}

\begin{proof}
	Observe that by definition of $\hat{V}$, $\hat{V}^\dagger \U V = \U$.
	Letting $\E = \oU - \U$ to simplify notation we thus have
	\begin{equation}
		\mathop{\mathbb{E}}_{V \sim d} \left[ | \hat{V}^\dagger \oU V \ket{\psi} - \U \ket{\psi} |^2 \right] = \mathop{\mathbb{E}}_{V \sim d} \left[ | \hat{V}^\dagger \E V \ket{\psi} |^2 \right].
	\end{equation}
	Observe that $\hat{V}^\dagger$ is unitary and thus does not affect the value of the norm, and can be dropped.
	Furthermore, note that the distribution over states $\ket{\phi} \equiv V\ket{\psi}$ forms a state 1-design.
    Denoting that state 1-design as $d'$ we have
	\begin{equation}
		\mathop{\mathbb{E}}_{V \sim d} \left[ | \hat{V}^\dagger \E V \ket{\psi} |^2 \right] = \mathop{\mathbb{E}}_{\ket{\phi} \sim d'} \left[ |\E \ket{\phi}|^2 \right] = \frac{\mathrm{Tr}\left[\E^\dagger \E \right]}{\dim{\H}}
	\end{equation}
	where the second equality comes from the definition of a state 1-design.
    Definition~\ref{def:opt-circuits-basis-independent} of an optimistic quantum circuit states that $\mathrm{Tr}\left[\E^\dagger \E \right]/\dim{\H} \leq \epsilon$, proving the theorem.
\end{proof}

As alluded to earlier, we may derandomize the reduction in Theorem~\ref{thm:randomized-reduction} via purification, by encoding the distribution of unitaries into the quantum state of a control register.
This idea is formalized in the following theorem.

\begin{theorem}[Derandomized reduction]
	\label{thm:purification}
	Consider an optimistic quantum circuit $\C$ with error parameter $\epsilon$, and let $\U$, $\oU$, and $\H$ be defined as in Definition~\ref{def:opt-circuits}.
	Now consider a unitary 1-design consisting of the uniform distribution over a discrete set of $k$ unitaries $\{V_i\}$ indexed by $i$, and let $V' = \sum_i \ket{i}\bra{i} \otimes V_i$.
	Then, for arbitrary states $\ket{\psi} \in \H$, the unitary $\U' = \hat{V}'^\dagger (\mathbb{I} \otimes \oU) V'$ applied to the state $\frac{1}{\sqrt{k}} \sum_i \ket{i} \otimes \ket{\psi}$ has error
	\begin{equation*}
		\left| (\U' - \mathbb{I} \otimes \U) \left(\frac{1}{\sqrt{k}} \sum_i \ket{i} \otimes \ket{\psi}\right) \right|^2 \leq \epsilon
	\end{equation*}
\end{theorem}

\begin{proof}
    We begin by understanding the action of the derandomized unitary $\mathcal{U}'$:
    \begin{equation}
        \U'\bigg(\sum_i \frac{\ket{i}}{\sqrt{k}}\otimes \ket{\psi}\bigg) = \sum_i \frac{\ket{i}}{\sqrt{k}}\otimes \hat{V}_i^\dagger\oU V_i\ket{\psi}
    \end{equation}
    Which then allows us to expand the error, and relate it to the randomized case:
    \begin{align}
    	&\left|(\U'-\mathbb{I} \otimes \U)
    	\left(\sum_i \frac{\ket{i}}{\sqrt{k}}\otimes\ket{\psi}\right)\right|^2
    	\notag\\
    	={}&\left|\sum_i \frac{\ket{i}}{\sqrt{k}}
    	\otimes(\hat{V}_i^\dagger\oU V_i-\U)\ket{\psi}\right|^2
    	\notag\\
    	={}&\mathop{\mathbb{E}}_{V\sim d}
    	\left[\,\big|\hat{V}^\dagger\oU V\ket{\psi}-\U\ket{\psi}\big|^2\,\right]
    	\notag\\
    	\leq{}&\;\epsilon
    \end{align}
    \noindent where in the last line, we leveraged Theorem~\ref{thm:randomized-reduction}.
\end{proof}

We remark that the notion of average error $\epsilon$ discussed above, implies an error $O(\sqrt{\epsilon})$ in diamond distance when viewed as a channel.

We conclude this subsection with a brief discussion comparing the randomized reduction of Theorem~\ref{thm:randomized-reduction} with the derandomized one of Theorem~\ref{thm:purification}. 
In practice, there is no reason to avoid the (cheaper) randomized construction: the classical control system of a real quantum computer would simply sample from the probability distribution during circuit compilation, inserting the resulting unitary in place of the channel.
(For a circuit that is used many times, this could be done independently each time).
We include the derandomized construction here simply because we find it an interesting theoretical fact that the reduction can be unitary, rather than only being possible as a non-unitary quantum channel.

\subsection{Warm-up example: optimistic quantum addition}
\label{subsec:opt-adder}

As a first example of an optimistic quantum circuit, we give an argument for why a low-depth circuit may approximate the addition unitary, for most inputs.%
\footnote{
	Note that we use quantum addition circuits in the reduction of Section~\ref{sec:oqft-reduction}. However in that case we cannot use the optimistic adders presented in this section, because the whole point of the reduction is that it should still work on \emph{worst-case} inputs.
}
Let the addition unitary be
\begin{equation*}
\add \ket{a}\ket{b}\ket{0} = \ket{a}\ket{s} \,,
\end{equation*}
where \(a\) and \(b\) are \(n\)-bit integers, and \(s = a+b\) is an \(n+1\) bit integer.  The construction employs some common notation used when describing addition circuits. Denote the bits of an integer \(x\) as \(x = x_{n-1}x_{n-2}\dots x_1x_0\), where \(x_0\) is the least-significant bit. The bits of \(s\) are then:
\begin{align}
  c_0 &= 0 \label{eqn:add-1}\\
  c_{i+1} &= \textsf{MAJ}(a_{i}, b_i, c_i) \text{ for } 0 \le i \le n-1 \label{eqn:add-2} \\
  s_i &= a_i \oplus b_i \oplus c_i \text{ for } 0 \le i \le n-1 \label{eqn:add-3}\\
  s_n &= c_n \label{eqn:add-4}
\end{align}
where the \textsf{MAJ} function equals one if the majority of its inputs equal one, and zero otherwise. The values \(c_i\) are referred to as the \emph{carries}.

As is well known, Equations \ref{eqn:add-1}-\ref{eqn:add-4} show that addition is \emph{non-local}, in that we cannot compute the output bit \(s_i\) just by looking at a local neighborhood of bits around index $i$. Instead, information is propagated from one end of the input all the way to the other via the carry bits. This non-locality is a feature that will show up again when we examine the QFT in the following section. Additionally, this non-locality places constraints on the possible circuits we can use to implement addition: it is impossible to achieve logarithmic depth for addition circuits without some kind of long-range information flow.  However, the passing of this information can be accomplished in many ways besides direct use of geometrically long-range gates, including through classical means, via measurement and feedforward operations \cite{baumer_approximate_2025}, such that quantum operations may remain more local.

The approach we take here is different, and is motivated by the intuition that only a small fraction of inputs cause addition to propagate information over a large distance. 
Thus, for an optimistic adder, the propagation of carries can simply be cut off after some number of bits.
Indeed, consider the logarithmic-depth, logarithmically-local adder of Algorithm~\ref{alg:oadd}.
The output sum $s$ is only incorrect if for at least one $i$, three conditions are met: all of the bits of $S_i$ (block $i$ of the sum $s$) are zero, the computed incoming carry $c_{i-1} = 1$, \textit{and} the computed outgoing carry $c_i = 0$.
In that case, the incoming carry caused block $i$ of sum to wrap around mod $2^m$, and the outgoing carry $c_i$ should have been set to 1---but the information from the incoming carry hadn't arrived yet.
Using just the requirement that some $m$-bit $S_i = 0$, we may bound the optimistic error for this circuit as $\epsilon \leq (n/m) / 2^m = O(1/\mathrm{poly}(n))$ since $m = O(\log n)$.

\begin{algorithm}
	\caption{(warm-up) An optimistic adder.}
	\label{alg:oadd}
	
	\KwIn{Quantum registers $\ket{a}$ and $\ket{b}$ with $n$ qubits each}
	
	Let $m = O(\log n)$ be the block size.
	
	\vspace{0.5em}
	
	\begin{enumerate}
		\item Break the qubits of $\ket{a}$ and $\ket{b}$ into blocks $\{\ket{A_i}\}$ (resp. $\{\ket{B_i}\}$) of $m$ qubits each, where $i=0$ corresponds to the least significant bits.
		\item Perform the following steps in parallel for all $i$:
		\begin{enumerate}
			\item Add $A_i$ in-place to $B_i$, storing the outgoing carry bit $c_i$ in an ancilla qubit: $\ket{A_i}\ket{B_i}\ket{0} \to \ket{A_i}\ket{A_i + B_i}\ket{c_i}$.
			\item Add the value of carry $c_i$ into block $i+1$ of the second (output) register, $\bmod{2^m}$: $\ket{A_i}\ket{B_i}\ket{c_i} \to \ket{A_i}\ket{A_i + B_i + c_i}\ket{c_i}$.
			\item Uncompute the carry qubits $\ket{c_i}$.\footnote{This requires a small trick given that $B_i$ is no longer available; see~\cite{draper_logarithmic-depth_2006}.}
		\end{enumerate}
	\end{enumerate}
\end{algorithm}

Because there already exist exact adders that simultaneously achieve low qubit count and depth~\cite{takahashi_fast_2008, remaud_ancilla-free_2025}, the benefit of using optimistic circuits in this simple example is essentially just locality of gates.
In the next section we describe an optimistic construction for the quantum Fourier transform that yields dramatic improvements over non-optimistic constructions in depth, locality, and qubit count, simultaneously.

\section{Ancilla-free logarithmic-depth optimistic QFT}
\label{sec:opt-qft}

\begin{figure*}
	\begin{center}
		\includegraphics{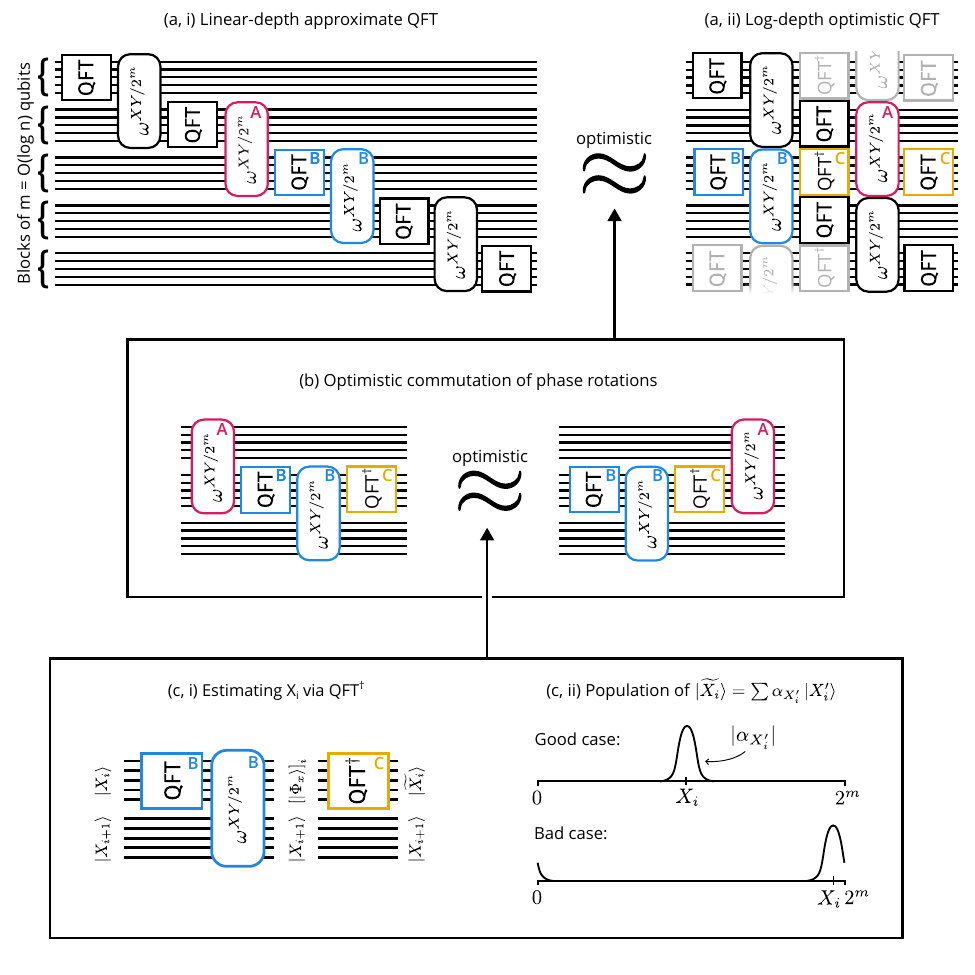}
	\end{center}
	\caption{
		\textbf{Concept behind the optimistic QFT.}
		Starting with a ``blocked'' version of a standard approximate QFT circuit (a,i), we add an identity of the form $\qft^\dagger \qft$ (labeled ``C'', yellow border) and then move approximately commuting blocks past each other (b) to yield the optimistic QFT (a,ii).
		The blocks approximately commute because applying $\qft^\dagger$ is equivalent to quantum phase estimation (c,i), which yields a superposition of values numerically close to the input state $\ket{X_i}$.
		The quantum phase estimation fails when $X_i$ is near 0 or $2^m$ because the superposition wraps around mod $2^m$ (c,ii).
	}
	\label{fig:oqft-detail}
\end{figure*}

In this section we describe in detail our construction for the
optimistic quantum Fourier transform (QFT).
We begin with some discussion of the structure of the QFT, and build off of that to define our optimistic QFT circuit, which we denote $\oqft$.
We prove that our construction is an optimistic circuit with error parameter $\epsilon$; we also explicitly identify the computational basis states which have large error (see Lemma~\ref{lem:large-error-subspaces}).  %

The quantum Fourier transform is typically expressed in terms of Fourier basis states:%
\footnote{In this work we exclusively discuss the QFT modulo $2^n$.}
\begin{equation} \label{eq:qft-math}
	\ket{\Phi_x} \equiv \qft \ket{x} = \sum_{y=0}^{2^n-1} e^{2\pi i xy/2^n} \ket{y}
\end{equation}
where $n$ is the number of qubits, and for simplicity we have dropped normalization (as we do throughout this section).
Perhaps surprisingly given the form of Equation~\ref{eq:qft-math}, Fourier basis states are product states; they can also be expressed thus:
\begin{equation} \label{eq:qft-product}
	\ket{\Phi_x} = \bigotimes_{i=0}^{n-1} \left( \ket{0} + e^{2\pi i \; 0.x_i x_{i+1} \cdots} \ket{1} \right)
\end{equation}
where $x_i$ is the $i^\mathrm{th}$ bit of $x$, and $0.x_i x_{i+1} \cdots = 2^i x / 2^n \bmod 1$ is a binary fraction
\cite{nielsen_quantum_2011}.%
\footnote{
	It is important to note that with the indexing of Equation~\ref{eq:qft-product}, $x$ and $y$ in Equation~\ref{eq:qft-math} have different endianness.
	Because it is natural for the circuits we discuss, we use that indexing throughout this work.
	If needed, the bit reversal unitary on $n$ qubits can be implemented with a single layer of $n/2$ \textsf{SWAP} gates.
}
It was observed early in the study of the QFT that an approximation of $\ket{\Phi_x}$ with fidelity $\epsilon$ can be achieved by simply truncating the binary fraction in each qubit's phase to $m = O(\log (n/\epsilon))$ bits~\cite{coppersmith_approximate_2002}.
This yields a standard construction for the approximate QFT: iterate through each qubit, applying a Hadamard gate to qubit $i$ and then a series of controlled phase rotations from qubits $i+1$ through $i+m-1$.
This construction requires only $O(n \log (n/\epsilon))$ gates, but has linear depth, because the controlled phase rotations on each qubit rely on the following qubits still being in their initial (input) state.

Equation~\ref{eq:qft-product} may seem to suggest that information remains local in the QFT, because each qubit's output state only depends on $O(\log (n/\epsilon))$ nearby qubits.
Perhaps the approximate QFT should thus be achievable via a low-depth circuit of entirely local gates.
Relatedly, it has recently been shown that the QFT has small entanglement~\cite{chen_quantum_2023}.
However, a simple counterexample shows that information from a small subset of the qubits can affect the entire output.
Consider the two states $\ket{\phi_0} = \qft \ket{000\cdots000}$ and $\ket{\phi_1} = \qft \ket{111\cdots111}$.
From Equation~\ref{eq:qft-product}, it is straightforward to see that $\ket{\phi_0}$ is simply a collection of $\ket{+}$ states.
Surprisingly, all but the last few qubits of $\ket{\phi_1}$ are exponentially close to the $\ket{+}$ state as well!
It is just the last $O(\log n)$ qubits that non-negligibly distinguish $\ket{\phi_0}$ from $\ket{\phi_1}$, so $\qft^\dagger$ must propagate the information from these qubits across the entire register to yield $\ket{000\cdots000}$ or $\ket{111\cdots111}$.
Thus via a light-cone argument, it is impossible to construct a circuit for the quantum Fourier transform that both has low depth and uses entirely local gates (in 1D).
The key to our optimistic QFT is to recognize that this long-range behavior is in some sense \emph{atypical}.
Indeed, we find that it \emph{is} possible to construct a low-depth, entirely local circuit which approximates the QFT well on the vast majority of basis states; we do so in the following.

Let us divide the $n$-qubit input register into blocks, each having $m = O(\log(n/\epsilon))$ qubits.%
\footnote{For simplicity, in this section we assume $n$ is a multiple of $m$; one may simply resize the last block to make the construction work with arbitrary $n$.}
For a computational basis state $\ket{x}$, we denote the $m$-bit integer value on block $i$ as $X_i$, such that $x = \sum_i 2^{mi} X_i$.
Now, we can write a blockwise version of Equation~\ref{eq:qft-product}:
\begin{equation} \label{eq:block-exact}
	\ket{\Phi_x} = \bigotimes_i \left[ \sum_{Y_i = 0}^{2^m - 1} \omega^{\sum_{j=i}^{n/m}(X_j/2^{m(j-i)}) Y_i} \ket{Y_i} \right]
\end{equation}
where $\omega = e^{2\pi i /2^m}$.
The choice of our block size $m = O(\log n/\epsilon)$ is convenient because it allows us to approximate this state to within $\epsilon$ by simply dropping all but the first two terms of the sum in the exponent.
Denoting the $i^\mathrm{th}$ block of $\ket{\Phi_x}$ as $\left[\ket{\Phi_x}\right]_i$, we have
\begin{equation} \label{eq:block-approx}
	\left[\ket{\Phi_x}\right]_i \approx \sum_{Y_i = 0}^{2^m - 1} \omega^{(X_i + X_{i+1}/2^m) Y_i} \ket{Y_i}.
\end{equation}
Equation~\ref{eq:block-approx} suggests a ``blockwise'' version of the standard circuit for computing the approximate $\qft$, shown in Figure~\ref{fig:oqft-detail}(a,i): iterate through the blocks, applying a local $m$-bit $\qft$ to block $i$ (yielding the state $\sum_{Y_i} \omega^{X_i Y_i} \ket{Y_i}$) and then a phase rotation of $\omega^{X_{i+1} Y_i /2^m}$ between neighboring blocks $i$ and $i+1$.
(See Appendix~\ref{app:aqft}).
This circuit has linear depth because the operations on block $i-1$ need access to $X_i$, so we cannot transform block $i$ of the input until block $i-1$ is complete.

In order to achieve low depth we must break this chain of dependencies, and we do so by attempting to estimate $X_i$ directly from $\left[\ket{\Phi_x}\right]_{i}$.
The study of quantum phase estimation (QPE) suggests that applying a length-$m$ inverse QFT to $\left[\ket{\Phi_x}\right]_i$ would be a promising strategy: we expect QPE to yield a superposition $\ket{\widetilde{X_i}}$ whose population is peaked on computational basis states $\ket{X_i'}$ having $X_i'$ close to the value of the phase factor $X_i + X_{i+1}/2^m$ in Equation~\ref{eq:block-approx} (see Nielsen \& Chuang, section 5.2.1~\cite{nielsen_quantum_2011}).
Thus, intuitively, we should be able to approximately apply the phase rotation $\omega^{X_i Y_{i-1} / 2^m}$ by using the state $\ket{\widetilde{X_i}}$, with only a small phase error.
Unfortunately, there is a subtle but critical problem: the distribution of population in $\ket{\widetilde{X_i}}$ wraps around modulo $2^m$~\cite{rall_faster_2021}.%
\footnote{Mathematically, for basis states $\ket{X_i'}$ with high weight, $|X_i'-X_i|_{2^m}$ is small. Here $|\cdot|_{2^m}$ is the Lee metric: $|z|_{2^m} = \min(z \bmod 2^m, -z \bmod 2^m)$ which is the definition of absolute value that is most sensible in this context.}
If $X_i$ is too close to 0 or $2^m$, a non-negligible portion of the distribution will be off by roughly $2^m$ and produce a large error in the phase (see Figure~\ref{fig:oqft-detail}(b,ii)).
The problem is not with our phase estimation procedure; it is simply impossible in general to estimate $X_i$ from only $\left[\ket{\Phi_x}\right]_i$ for all $x$.
This is where we can use the power of optimistic circuits.
We may simply \emph{allow} some states to have large error, observing that because the distribution in $\ket{X_i'}$ is highly peaked, for \emph{most} inputs the error from wraparound will be small, and thus the average error over basis states will be small.

\begin{algorithm}
	\caption{The optimistic quantum Fourier transform.}
	\label{alg:oqft}

	\KwIn{Quantum register with $n$ qubits; error parameter $\epsilon$}

	Let $m = O(\log n/\epsilon)$ be the block size.

	\begin{enumerate}
		\item Apply $\qft_{2^m}$ to the even-indexed blocks.
		\item Apply the phase rotation $\exp{(2\pi i X_{i+1} Y_i / 2^{2m})}$ \linebreak for all even $i$.
		\item Apply $\qft_{2^m}$ to the odd-indexed blocks, and apply $\qft_{2^m}^\dagger$ to the even-indexed blocks.
		\item Apply the phase rotation $\exp{(2\pi i X_{i+1} Y_i / 2^{2m})}$ \linebreak for all odd $i$.
		\item Apply $\qft_{2^m}$ to the even-indexed blocks.
	\end{enumerate}
\end{algorithm}

Following this intuition, we may construct the optimistic QFT circuit, shown in Figure~\ref{fig:oqft-detail}(a,ii) and Algorithm~\ref{alg:oqft}.
We first compute $\left[\ket{\Phi_x}\right]_i$ on blocks with even $i$, which is straightforward to do in parallel because the odd-indexed blocks remain in their initial state.
We then apply $\qft^\dagger$ to the even blocks, yielding $\Xp$ on those blocks {[}Fig.~\ref{fig:oqft-detail}(b,i){]}, which allows us to compute $\left[\ket{\Phi_x}\right]_i$ for odd $i$.
Importantly, when the phase rotation $\omega^{X_{i+1}' Y_i / 2^m}$ is a good approximation of $\omega^{X_{i+1} Y_i / 2^m}$ there is negligible phase kickback to $\Xp$ and the state on the blocks with even $i$ is approximately unchanged.
Thus a final $\qft$ applied to the even blocks returns them to $\left[\ket{\Phi_x}\right]_i$, yielding the desired approximation of $\Phi_x$ on the entire register.
In the following theorem we formalize that it is an optimistic circuit for the $\qft$ with error parameter $\epsilon$.
\begin{restatable}[Optimistic quantum Fourier transform]{theorem}{oqftthm}
	\label{thm:oqft}
	The circuit defined in Algorithm~\ref{alg:oqft} is an optimistic circuit with error parameter $\epsilon$ for $\qft_{2^n}$.
\end{restatable}

\begin{proof}
	See Appendix~\ref{app:oqft-proof}.
\end{proof}

\begin{figure}[t]
	\begin{center}
		\includegraphics[width=0.7\linewidth]{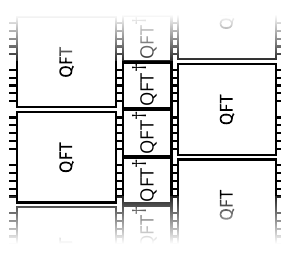}
	\end{center}
	\caption{An alternative way of expressing the optimistic QFT circuit, see Algorithm~\ref{alg:oqft-qft-only}. Equivalent to Figure~\ref{fig:oqft-detail}(a,ii) but written using only QFTs on blocks of $O(\log(n/\epsilon))$ qubits (see Appendix~\ref{app:oqft-qft-only}).}
	\label{fig:oqft-qft-only}
\end{figure}

\begin{algorithm}[t]
	\caption{The optimistic quantum Fourier transform (alternate form using only QFT blocks).}
	\label{alg:oqft-qft-only}

	\KwIn{Quantum register with $n$ qubits; error parameter $\epsilon$}

	Let $m = O(\log n/\epsilon)$ be the block size.

	\begin{enumerate}
		\item Apply $\qft_{2^{2m}}$ spanning blocks $i$ and $i+1$, for all even $i$.
		\item Apply $\qft_{2^m}^\dagger$ to all blocks.
		\item Apply $\qft_{2^{2m}}$ spanning blocks $i$ and $i+1$, for all odd $i$.
	\end{enumerate}
\end{algorithm}

We also present an alternative form of the optimistic QFT, where the phase rotation blocks have been subsumed into larger QFTs of size $2m$ (see Algorithm~\ref{alg:oqft-qft-only} and Figure~\ref{fig:oqft-qft-only}).
We explicitly show the equivalence of the circuits in Appendix~\ref{app:oqft-qft-only}.
This form is convenient for concrete implementations because any exact QFT circuit can be used to implement the blocks and no other gates need to be specified.
It is also conceptually simpler, although if implemented naively it will include gates that immediately cancel each other out and thus should be removed.
Some small phase rotations in the $2m$-qubit QFTs can also be dropped while maintaining error $\epsilon$.

We now discuss the cost of the optimistic QFT, using the alternative formulation (Fig.~\ref{fig:oqft-qft-only} and Alg.~\ref{alg:oqft-qft-only}) because its analysis is simpler.
Implementing the QFT blocks using the standard, linear-depth exact QFT construction~\cite{nielsen_quantum_2011} yields circuit depth $O(m)=O(\log (n/\epsilon))$.
It requires no ancilla qubits, no measurements, and the gates have range at most $2m = O(\log (n/\epsilon))$ for qubits arranged in 1D.
The locality can be improved to nearest-neighbor without ancillas by the insertion of some swap gates inside the exact QFT blocks~\cite{fowler_implementation_2004}; the asymptotic gate count and depth of the circuit are maintained.%
\footnote{The construction of~\cite{fowler_implementation_2004} reverses the order of the qubits; the qubit order can be reversed back in depth $O(m)$ via swap gates.}

Before moving on, we note an interesting fact about the ``bad'' subspace on which our optimistic QFT exhibits large error.
Earlier it was shown via light-cone argument that achieving low error on all inputs requires either long-range interactions (perhaps via measurement and feedforward), or large depth.
The counterexamples $\ket{000\cdots000}$ and $\ket{111\cdots111}$ from that discussion, in which the QFT transmits information at long range, are precisely the type of states in which the phase estimation trick in our optimistic QFT breaks down!
More precisely, the computational basis states in which our construction has large error are those for which at least one even-indexed $m$-qubit block $\ket{X_i}$ starts with a long string of 0s or 1s (and thus $X_i$ is close to $0$ or $2^m$).
One may wonder whether these states are somehow important for the QFT's power as a building block for quantum algorithms.
In the next section, we show that this is not the case, by using the optimistic QFT to construct an optimistic circuit for integer multiplication and showing that that circuit is sufficient for quantum integer factorization via Shor's algorithm.

\section{Optimistic multiplication and
factoring}\label{sec:opt-qft-factoring}

Integer multiplication is a core operation in a wide array of quantum algorithms.
In this section, we consider the setting of Shor's algorithm for factoring an $n$-bit integer $N$~\cite{shor_polynomial-time_1997}.
The standard factoring circuit repeatedly invokes in-place modular multiplication of a quantum register by a classical constant $c$, where $c$ is coprime to $N$:
\begin{equation}
	\label{eq:mod-mult}
	\ket{y} \to \ket{cy \bmod N}  \,.
\end{equation}
It has been shown that this operation can be performed via four QFTs, in addition to other operations which use $2n + O(n/\log n)$ total qubits and have total depth $O(n^\delta)$ for a tunable $\delta > 0$ (note that the constant factors hidden by the asymptotic notation become larger as $\delta$ becomes closer to zero)~\cite{kahanamoku-meyer_fast_2024}.
Surprisingly, the bottleneck is the QFTs: all previously-known QFT constructions would dominate either the depth or qubit count of this multiplier.
Our QFT constructions completely address this bottleneck, yielding a multiplier with depth $O(n^\delta)$ and $2n + O(n/\log n)$ total qubits.
This depth and qubit count is immediate using our randomized approximate QFT from Section~\ref{sec:oqft-reduction}; the result is a factoring algorithm with depth $O(n^{1+\delta})$ yet only $2n + O(n/\log n)$ total qubits, by far the best known depth for a factoring circuit using so few qubits.
For more details of how this work fits into the landscape of quantum factorization algorithms, see Appendix~\ref{app:previous-factoring}.

It turns out that for Shor's algorithm the reduction from Section~\ref{sec:oqft-reduction} is not even necessary: as an example of the applicability of optimistic circuits, we show in Appendix~\ref{app:shor-oqft-proof} that the optimistic QFT can be used directly in the multiplication subroutines to create an optimistic multiplier, and the probability of successfully finding the factors will remain high.
We remark that for the final QFT in Shor's algorithm (the one used to find the period), it is probably not the right choice to use the optimistic QFT.
Because that final QFT only happens once, using a linear-depth QFT there adds negligible depth cost to the overall period finding algorithm, the rest of which has super-linear depth.
Meanwhile, there is a concrete advantage to the standard QFT construction in this context: one qubit can be recycled for the entire $\ket{x}$ register, which is crucial for achieving qubit count $2n + o(n)$~\cite{beauregard_circuit_2003}. 

We end this section with a brief discussion for readers particularly interested in the practical cost of quantum algorithms for factoring and discrete logarithm.
In the near term, it seems that our ability to execute cryptographically relevant quantum algorithms will be limited by logical qubit count, with circuit depth and gate count being less relevant; in the medium to long term, circuit \textit{space-time volume}---space times depth---will become the important metric, as users aim for higher throughput on high-quality quantum devices.
In that era there may also be real-time applications for which circuit depth itself (which determines wall-clock time) is the relevant metric.
We believe that the constructions described in this section can concretely improve both space-time volume and depth, and thus will be relevant as the cost metrics change.
While our constructions do use very few ancilla qubits and thus save on space, it has recently been shown that it is possible to reduce even the \textit{non-ancilla} qubit count of quantum algorithms for factoring and discrete logarithm, achieving a total logical qubit count as low as $n/2 + o(n)$ at the expense of considerably higher gate count and depth~\cite{chevignard_reducing_2024, gidney_how_2025}. 
These algorithms achieve such dramatic space savings by eschewing large multiplications, and thus it seems unlikely that the optimistic fast multiplier described here would speed them up.
For the very first demonstrations of cryptographically relevant quantum factoring and discrete logarithm, logical qubit count seems to be so important that the higher gate count and depth will be worth the trade-off.

\section{Approximate QFT via reduction}
\label{sec:oqft-reduction}

We now apply the reduction described in Section~\ref{subsec:reduction} to the optimistic QFT, yielding approximate QFT circuits that have small error on arbitrary inputs.
We present two variations of the reduction, both maintaining depth $O(\log (n/\epsilon))$ which is asymptotically optimal (for measurement-free circuits~\cite{cleve_fast_2000}): a randomized construction using a sublinear number of ancillas, and a derandomized (unitary) one using $3n+o(n)$ ancillas.
The randomized construction seems to be the first approximate QFT to achieve this depth using a sublinear number of ancilla qubits; the purified construction similarly seems to be the first \textit{deterministic} circuit (no randomization, nor measurements and feed-forward) to achieve this depth using $3n + o(n)$ total qubits~\cite{cleve_fast_2000, hales_quantum_2002, baumer_approximate_2025}.

For the 1-design $d$ in the reduction (see Section~\ref{subsec:reduction}), we uniformly sample two random integers $r_1$ and $r_2$ between 0 and $2^n-1$, and then choose $V(r_1, r_2) = X_{2^n}^{r_1} Z_{2^n}^{r_2}$, where $X_{2^n}$ and $Z_{2^n}$ are the Weyl-Heisenberg generalization of the $X$ and $Z$ Pauli operators to a Hilbert space of dimension $2^n$.
Mathematically, $d$ is the uniform distribution over the Weyl-Heisenberg group, which is well known to form a 1-design~\cite{graydon_clifford_2021}.
For intuition, $X_{2^n}^{r_1} \ket{x} = \ket{(x+r_1) \bmod 2^n}$ corresponds to the addition of the integer $r_1$ into the register, and $Z_{2^n}^{r_2}$ corresponds to a ``phase gradient'' rotation $Z_{2^n}^{r_2} \ket{x} = e^{2\pi i r_2 x / 2^n} \ket{x}$.
Altogether, the action of $V(r_1, r_2)$ on a computational basis state is
\begin{equation}
	V(r_1, r_2) \ket{x} = e^{2\pi i r_2 x / 2^n} \ket{(x + r_1) \bmod 2^n}.
	\label{eq:weyl-pauli}
\end{equation}
In the case of the optimistic QFT, this choice of 1-design is particularly nice because the unitary $\hat{V}^\dagger = \qft \cdot V^\dagger \cdot \qft^\dagger$ (see Section~\ref{subsec:reduction}) is essentially the same operator, with the arguments exchanged:
\begin{equation}
	\hat{V}^\dagger (r_1, r_2) = V(r_2, -r_1)
	\label{eq:weyl-conj}
\end{equation}
(This is a direct result of the fact that $\qft$ is the Weyl-Heisenberg generalization of the Hadamard gate).
Applying Theorem~\ref{thm:randomized-reduction}, the optimistic QFT together with this 1-design yields a randomized approximate QFT with expected error $\epsilon$ for any input state, with the expectation taken over the randomness in the 1-design.

In the abstract circuit model, the phase rotation $Z_{2^n}^{r_2}$ can be implemented via one layer of single-qubit phase rotations.
Thus only the implementation of the integer addition $X_{2^n}^{r_1}$ may contribute non-negligibly to the overall circuit cost.
We formalize this discussion in the following theorem:

\begin{theorem}
    [The Randomized Approximate $\qft$]
	\label{thm:rand-qft}
	There exists a randomized circuit for the approximate quantum Fourier transform on $n$ qubits with gate cost $O(n \log (n/\epsilon))$, depth $O(\log (n/\epsilon))$, and ancilla qubit count $O(n / \log (n/\epsilon))$, where $\epsilon$ is the expectation value of the error on an arbitrary input state.
\end{theorem}
\begin{proof}
    We instantiate the circuit described in the previous paragraphs, using the classical-quantum adder of~\cite{takahashi_fast_2008}
    with block size $m = O(\log(n/\epsilon))$ for the integer addition $X_{2^n}^{r_1}.$
    \footnote{We note that the adder of~\cite{takahashi_fast_2008} makes use of the QFT adder~\cite{draper_addition_2000} as a subroutine, which requires many rotation gates. It is possible to instead use the linear-depth classical-quantum adder of~\cite{haner_factoring_2017}, resulting in a full, log-depth classical-quantum adder that only uses classically reversible gates. Furthermore, using the optimized circuits for Toffoli ladders from~\cite{remaud_ancilla-free_2025}, it is possible to reduce the ancilla count to $O \left (n/2^{\sqrt{\log (n/\epsilon)}}\right)$ by picking a block size of $m = 2^{\sqrt{\log (n/\epsilon)}}$.}
\end{proof}

Finally, we may apply the purification of Theorem~\ref{thm:purification}.
The qubit cost of the purification is large enough that the cost of the adders in the application of $V$ can be made negligible, yielding Theorem~\ref{thm:purified-approx-qft}.

\begin{restatable}[The Unitary Approximate $\qft$]{theorem}
{paqft}
	\label{thm:purified-approx-qft}
	There exists a unitary circuit for the approximate quantum Fourier transform on $n$ qubits with error $\epsilon$, with gate count $O(n \log (n/\epsilon))$, depth $O(\log (n/\epsilon))$, and total qubit count $3n+O(n/\log (n/\epsilon))$.
\end{restatable}
\begin{proof}
    This follows from implementing the controlled Weyl-Heisenberg operators using the quantum-quantum adder of~\cite{takahashi_quantum_2009} for $X_{2^m}^{r_1}$, and Lemma~\ref{lem:approxpg} (below) for $Z_{2^m}^{r_2}$. We defer a more detailed description to Appendix~\ref{sec:sec-5-pfs}.
\end{proof}

\begin{restatable}
	[Cost of approximate controlled phase gradients]{lemma}
	{approxpg}
	\label{lem:approxpg}
	There is a quantum circuit of depth $O(\log(n/\delta))$ that approximates the operation
	\[\ket{x}\ket{z} \mapsto e^{2\pi i x z / 2^n} \ket{x}\ket{z}\]
	to error $\delta$ in the operator norm.
\end{restatable}

\begin{proof}
	See Appendix~\ref{sec:sec-5-pfs}.
\end{proof}

Before concluding, we note an interesting connection to existing work on \emph{verification} of the approximate quantum Fourier transform.
It was previously shown that given black-box access to an approximate quantum Fourier transform with unknown error characteristics, it is possible to verify that the black box is sufficient for quantum phase estimation with only a polynomial number of calls to the black box~\cite{linden_average-case_2022}.
The strategy is to first verify that the error is well behaved on random inputs, and then when doing phase estimation apply a phase rotation $Z_{2^m}^r$ for some random $r$, and then classically subtract the value $r$ from the phase estimation result.
This corresponds to a special case of our reduction: if the output state will be measured immediately after the application of the QFT, the application of $\hat{V}^\dagger$ can be replaced with a classical subtraction.
Furthermore, the $X_{2^m}^{r_1}$ part of $V$ only results in a phase on the output states, which doesn't affect the results of the subsequent measurement, and thus in this context $X_{2^m}^{r_1}$ does not need to be performed, leaving only $Z_{2^m}^{r_2}$ as in the verification protocol.
The fact that the error of the black box is well-behaved on random inputs can be thought of as verifying that the black box performs an optimistic QFT; indeed, our optimistic QFT circuit can be thought of as \emph{instantiating} exactly such a QFT.

\section{Discussion and outlook}\label{sec-discussion}

Standard techniques from classical computing do not always map well onto quantum computing.
This is true for classical amortized algorithms, which have better average- than worst-case complexity but no obvious mapping to quantum circuits.
On the other hand, quantum computing provides us with new and different ways of solving algorithmic problems, which we can leverage to get around such obstacles.
In this work we introduce a new way of designing fast quantum circuits, via the use of circuits that approximate a desired unitary well on most, but not all, input states.
We demonstrate the power of our technique through the construction of a circuit for the approximate quantum Fourier transform.

Conceptually, we build on a long history of insight  \cite{mitzenmacher_probability_2005, motwani_randomized_1995} gained from randomized algorithms for classical computers.  Randomization is now being recognized as a powerful concept for quantum algorithms as well, including acceleration of quantum simulation \cite{campbell_shorter_2017,campbell_random_2019,nakaji_high-order_2024,childs_faster_2019,zhao_hamiltonian_2022,chen_average-case_2024,chen_concentration_2021}, other quantum algorithmic building blocks~\cite{martyn_halving_2025, dalzell_distillation-teleportation_2025, gosset_multi-qubit_2025}, and even quantum phase estimation \cite{linden_average-case_2022}.  The potential for randomized quantum algorithms is especially intriguing because coherent errors can add up differently from incoherent ones, and crossing back and forth across quantum and classical boundaries can open new opportunities for algorithmic cost reductions.   Mathematically, it is useful to keep in mind regarding quantum states and operators that the expected error of an approximate unitary applied to a random state is proportional to the Frobenius norm of the difference.  Optimistic quantum circuits are one direction in which to utilize these ideas.

Considering the wealth of classical algorithms having better average- than worst-case complexity, we believe that there may be a wide range of quantum circuits that can be optimized with the technique of optimistic quantum circuits.
One case in which its application may be particularly straightforward is in the design of ``semiclassical'' circuits, in which a classical function is applied to a superposition of inputs.
One example is the optimistic adder of Section~\ref{subsec:opt-adder}.
Another example is the computation of modular inverses, which are commonly used in quantum algorithms for cryptography and number theory problems~\cite{proos_shors_2003, roetteler_quantum_2017, haner_improved_2020}.
The standard classical algorithms for this problem (the Euclidean algorithm and binary GCD) have the frustrating property that their runtime depends on the inputs in an unpredictable way.
Optimistic quantum circuits may provide a way of achieving good performance for these algorithms, by simply accepting an incorrect answer on the inputs on which the algorithm does not terminate quickly!

A number of open questions stem from our optimistic QFT construction as well.
In this work we do not extensively explore the circuit's practical implementation; there are many circuit optimizations that can be applied, such as using phase gradient states to implement the local QFTs on blocks of $m$ qubits~\cite{kitaev_classical_2002, gidney_turning_2016, gidney_halving_2018, nam_approximate_2020}.
The parameters can also be tuned, for example the block size could differ between the odd and even blocks because they contribute to the error in different ways.
It would also be interesting to explore how the locality of the optimistic QFT circuit can be used to reduce routing cost in a practical factoring circuit.
Finally, an exciting direction for future research is the study of how optimistic circuits can inform our understanding of the fundamental characteristics of certain unitaries, and vice versa.
For example, the existence of a logarithmic-depth, logarithmically-local optimistic QFT is intimately related to the fact that the exact QFT has low entanglement \cite{chen_quantum_2023}; it would be interesting to explore whether optimistic circuits can teach us similar lessons about other unitaries.

~\\
\noindent
{\sf\large Contributions and acknowledgements}
~\\

\noindent
GKM created the optimistic QFT circuit;
CG pointed out all log-depth approximations of the QFT using only local gates must fail badly on certain inputs.
All authors contributed ideas to the construction of the worst-to-average-case reduction techniques; JB and TB devised the proofs of the error bounds for the reductions.
GKM wrote the first draft of the manuscript; all authors contributed to editing, improving, and polishing the manuscript.

GKM and ILC were supported in part by the Co-design Center for Quantum Advantage, funded by the Department of Energy as part of the National Quantum Initiative.  GKM also acknowledges support in part by the generosity of Google.
JB acknowledges support of the Doc Bedard Fellowship of the Laboratory for Physical Sciences. This work was completed in part while TB was a student researcher at Google.
We thank John Martyn for useful discussions regarding error bounds for randomized circuits, and Thomas Schuster for insights regarding $t$-designs.

\onecolumn

\printbibliography

\clearpage

\appendix{}

\section*{Appendix}\label{appendix}

\section{Comparison of previous constructions}

\subsection{Quantum Fourier transforms}
\label{app:previous-qft}

\begin{table*}[h]
	\begin{center}
		\begin{tabular}{|c|c|c|c|c|c|c|}
			\hline
			\textbf{Construction} & \textbf{Approx.} & \textbf{Gates} & \textbf{Depth} & \textbf{Ancillas} & \makecell{\textbf{Gate} \\ \textbf{range}} & \makecell{\textbf{Uses} \\\textbf{measmts.}} \\
			\hline \hline
			\multicolumn{7}{|c|}{\textbf{Exact constructions}}\\
			\hline
			Standard~\cite{nielsen_quantum_2011} & Exact & $O(n^2)$ & $O(n)$ & 0 & all-to-all & N \\
			\hline
			Few gates\textsuperscript{*}~\cite{cleve_fast_2000} & Exact & $O(n \log^2 n)$ & $O(n)$ & $O(n \log n)$ & all-to-all & N \\
			\hline
			Few gates, no ancillas\textsuperscript{\textdagger}~\cite{kahanamoku-meyer_fast_2024} & Exact & $O(n^{1+\delta})$ & $O(n)$ & 0 & all-to-all & N \\
			\hline
			Nearest-neighbor \cite{fowler_implementation_2004} & Exact & $O(n^2)$ & $O(n)$ & 0 & 1 & N \\
			\hline \hline
			\multicolumn{7}{|c|}{\textbf{Approximate constructions}}\\
			\hline
			Standard approx. \cite{coppersmith_approximate_2002, nielsen_quantum_2011} & Approximate & $O(n \log n)$ & $O(n)$ & 0 & $O(\log n)$ & N \\
			\hline
			Low depth~\cite{cleve_fast_2000} & Approximate & $O(n \log n)$ & $O(\log n)$ & $O(n \log n)$ & all-to-all & N \\
			\hline
			Improved low depth~\cite{hales_quantum_2002} & Approximate & $O(n \log n)$ & $O(\log n)$ & $O(n)$ & all-to-all & N \\
			\hline \hline
			\multicolumn{7}{|c|}{\textbf{Dynamic circuits (i.e. using measurement and classical feed-forward)}}\\
			\hline
			Constant depth~\cite{hoyer_quantum_2005, buhrman_state_2024} & Approximate & $O(n^3 \log n)$ & $O(1)$ & $O(n^3 \log n)$ & all-to-all & Y \\
			\hline
			Low-depth local~\cite{baumer_approximate_2025} & Approximate & $O(n \log n)$ & $O(\log n)$ & $O(n)$ & 1 & Y \\
			\hline \hline
			\multicolumn{7}{|c|}{\textbf{This work}}\\
			\hline	
			Optimistic QFT (Alg.~\ref{alg:oqft}) & Optimistic & $O(n \log n)$ & $O(\log n)$ & 0 & 1 & N \\
			\hline	
			Randomized reduction (Thm.~\ref{thm:rand-qft}) & Approximate & $O(n \log n)$ & $O(\log n)$ & $O(n/\log n)$ & all-to-all & N \\
			\hline
		\end{tabular}
	\end{center}
	
	\caption{Asymptotic circuit sizes for quantum Fourier transform constructions on $n$ qubits (with modulus $2^n$). For approximate and optimistic constructions, we take the error $\epsilon = O(1/\mathrm{poly}(n))$. ``Gate range'' is for input qubits arranged in 1D. We note that both~\cite{hales_quantum_2002} and~\cite{baumer_approximate_2025} discuss ``optimistic'' versions of their circuits (using different language), but the cost is only reduced by constant factors and thus the asymptotic scaling is the same as that of the approximate circuits listed in the table. \textsuperscript{*}The values listed here use the classical multiplication algorithm of~\cite{harvey_integer_2021}, which did not yet exist when~\cite{cleve_fast_2000} was published. \textsuperscript{\textdagger}$\delta$ is a positive constant which can be made arbitrarily close to zero.}
\end{table*}

\subsection{Quantum factoring circuits}
\label{app:previous-factoring}

\renewcommand{\arraystretch}{1.25}
\begin{table*}[h]
	\begin{center}
		\begin{tabular}{|c|c|c|c|}
			\hline
			\textbf{Construction} & \textbf{Gates} & \textbf{Depth} & \textbf{Qubits} \\
			\hline
			\hline
			Shor\textsuperscript{*}~\cite{shor_polynomial-time_1997} & $\widetilde{O}(n^2)$  & $\widetilde{O}(n^2)$ & $\widetilde{O}(n)$ \\
			\hline
			Beauregard~\cite{beauregard_circuit_2003} & $\widetilde{O}(n^3)$ & $O(n^3)$ & $2n + 3$ \\
			\hline
			Zalka~\cite{zalka_shors_2006} & $\widetilde{O}(n^3)$ & $O(n^3)$ & $\frac{3}{2} n + 2$ \\
			\hline
			Quantum Karatsuba~\cite{gidney_asymptotically_2019} & $O(n^{2.58 \cdots})$ & $\widetilde{O}(n^{1.58 \cdots})$ & $O(n)$ \\
			\hline
			RNS-style~\cite{chevignard_reducing_2024} & $\widetilde{O}(n^3)$ & $\widetilde{O}(n^3)$ & $n/2 + o(n)$ \\
			\hline
			\textbf{This work + \cite{kahanamoku-meyer_fast_2024}\textsuperscript{\textdagger}} & $O(n^{2 + \delta})$ & $O(n^{1 + \delta'})$ & $2n + o(n)$ \\
			\hline
		\end{tabular}
	\end{center}
	
	\caption{Circuit costs for selected quantum circuits to factor $n$-bit RSA integers. \textsuperscript{*}The values listed here use the classical multiplication algorithm of~\cite{harvey_integer_2021}, which did not yet exist when~\cite{shor_polynomial-time_1997} was published. \textsuperscript{\textdagger}$\delta$ and $\delta'$ are arbitrary positive constants which can be made arbitrarily close to zero.}
\end{table*}

\section{Linear-depth approximate QFT}
\label{app:aqft}

In the construction of the optimistic QFT, we reference an approximate QFT construction in which operations are applied on blocks of qubits.
In Algorithm~\ref{alg:aqft} we record this construction explicitly, and in Lemma~\ref{lem:aqft-error} we bound its error.

\begin{algorithm}
	\caption{The linear-depth approximate quantum Fourier transform (see Figure~\ref{fig:oqft-detail}(a)).}
	\label{alg:aqft}
	
	\KwIn{Quantum register with $n$ qubits; block size $m$}
	
	\vspace{1em}
	
	Let $\imax = \lceil n/m \rceil$.  For $i$ from 0 to $\imax-1$:
	\vspace{1em}
	\begin{enumerate}
		\item Apply $\qft_{2^m}$ to block $i$.
		\item Apply the phase rotation $\exp{(2\pi i X_{i+1} Y_i / 2^{2m})}$, where $X_{i+1}$ and $Y_i$ are the values of blocks $i+1$ and $i$ respectively as integers in the computational basis.
	\end{enumerate}
\end{algorithm}

\begin{lemma}[Frobenius error of linear-depth approximate QFT]
	\label{lem:aqft-error}
	Let $\qft'$ be the unitary induced by Algorithm~\ref{alg:aqft} with block size $m$, and let $\qft$ be the exact QFT.
	Then
	\begin{equation}
		\frac{1}{2^n} ||\qft' - \qft||_F^2 \leq \frac{4 \pi^2}{3} \cdot \left\lceil \frac{n}{m} \right\rceil \cdot \frac{1}{2^{m}}
	\end{equation}
	(cf. Definition~\ref{def:opt-circuits-basis-independent}).
\end{lemma}

\begin{proof}
	First, note that
	\begin{equation}
		||\qft' - \qft||_F^2 = \sum_{x=0}^{2^n-1} |\qft' \ket{x} - \qft \ket{x} |^2
	\end{equation}
	for computational basis states $\ket{x}$ (see Def.~\ref{def:opt-circuits}).
	The result of applying Algorithm~\ref{alg:aqft} to such a state $\ket{x}$ is given in Equation~\ref{eq:block-approx}, reproduced (in slightly different form) here:
	\begin{equation}
		\qft'\ket{x} = \bigotimes_{i=0}^{\imax} \sum_{Y_i = 0}^{2^m - 1} \omega^{(X_i + X_{i+1}/2^m) Y_i} \ket{Y_i}
	\end{equation}
	where $X_i$ is the $i^\mathrm{th}$ block of $x$ and $\omega = \exp{2\pi i/2^m}$.
	The action of the exact QFT is
	\begin{equation}
		\qft\ket{x} = \bigotimes_{i=0}^{\imax} \sum_{Y_i = 0}^{2^m - 1} \omega^{\sum_{j=i}^{\imax}(X_j/2^{m(j-i)}) Y_i} \ket{Y_i}
	\end{equation}
	Thus we have
	\begin{equation}
		||\qft' - \qft||_F^2 = \sum_x \sum_i \sum_{Y_i = 0}^{2^m - 1} \left| \omega^{(X_i + X_{i+1}/2^m) Y_i} -  \omega^{\sum_{j=i}^{n/m}(X_j/2^{m(j-i)}) Y_i} \right|^2
	\end{equation}
	Using the definition $\omega = e^{2\pi i/2^m}$ and the fact that $|e^{i \alpha} - e^{i \beta}| \leq |\alpha - \beta|$, we have
	\begin{equation}
		||\qft' - \qft||_F^2 \leq \sum_x \sum_i \sum_{Y_i = 0}^{2^m - 1} \left[ \frac{2 \pi}{2^m} \sum_{j=i+2}^{n/m} (X_j/2^{m(j-i)}) Y_i \right]^2
	\end{equation}
	The $X_j$ are all at most $2^m - 1$, so we can bound $\sum_{j=i+2}^{n/m} X_j/2^{m(j-i)} \leq 1/2^m$, yielding
	\begin{equation}
		||\qft' - \qft||_F^2 \leq \sum_x \sum_i \sum_{Y_i = 0}^{2^m - 1} \left[ \frac{2 \pi}{2^{2m}} Y_i \right]^2 \leq 2^n \cdot \imax \cdot \frac{4 \pi^2}{2^{4m}} \cdot \frac{2 \cdot 2^{3m}}{6}
	\end{equation}
	or, finally,
	\begin{equation}
		\frac{1}{2^n} ||\qft' - \qft||_F^2 \leq \frac{4 \pi^2}{3} \cdot \left\lceil \frac{n}{m} \right\rceil \cdot \frac{1}{2^{m}}
	\end{equation}
	which is what we desired to show.
\end{proof}

\section{Proof of optimistic QFT}
\label{app:oqft-proof}

Here we provide a proof of Theorem~\ref{thm:oqft}, restated here:

\oqftthm*{}

\begin{proof}
Following Definition~\ref{def:opt-circuits-basis-independent}, we desire to show that
\begin{equation}
	\frac{1}{2^n} ||\oqft - \qft||_F^2 \leq \epsilon
\end{equation}

Consider the linear-depth approximate QFT on block size $m$ described in Appendix~\ref{app:aqft}.
By Lemma~\ref{lem:aqft-error}, it induces a unitary $\qft'$ for which
\begin{equation}
	\frac{1}{2^n} ||\qft' - \qft||_F^2 \leq \frac{4 \pi^2}{3} \cdot \left\lceil \frac{n}{m} \right\rceil \cdot \frac{1}{2^{m}}
\end{equation}
Thus we proceed by bounding $||\oqft - \qft'||_F^2$, which we may use to prove the theorem via the triangle inequality.
In particular, we show that the linear-depth approximate QFT can be transformed into the optimistic QFT by commuting certain gates past each other, and then in Lemma~\ref{lem:commutation-error} show that doing so maintains a small Frobenius norm.

The transformation is depicted graphically in Figure~\ref{fig:oqft-detail}.
Starting with the linear-depth approximate QFT (see Appendix~\ref{app:aqft}), we first add a resolution of the identity of the form $\qft \cdot \qft^\dagger$ to the end of the circuit on each even-indexed block.
Now, on every even-indexed block $i$ is the sequence of gates depicted in Figure~\ref{fig:commutation-lemma}(a): a phase rotation between blocks $i-1$ and $i$, $\qft$ on block $i$, a phase rotation between blocks $i$ and $i+1$, and finally $\qft^\dagger$ on block $i$.
We denote this block of four gates as $W$.
The key transformation is to move the first phase rotation past the other gates to the end of the sequence, producing a new sequence $W'$ (Fig.~\ref{fig:commutation-lemma}(b)).

Let $\oqft_j$ be the ``intermediate'' circuit in which $j$ of the blocks $W$ in the linear-depth QFT have been replaced with $W'$.
Then Lemma~\ref{lem:commutation-error} implies that
\begin{equation}
    \frac{1}{2^n} ||\oqft_i - \oqft_{i-1}||_F^2 \leq \frac{\Omega(m)}{2^{m}}
\end{equation}
Then, by triangle inequality, making all $\lceil n/m\rceil$ replacements yields an error
\begin{equation}
    \frac{1}{2^n} ||\oqft - \qft'||_F^2 \leq \frac{n^2}{m^2} \cdot \frac{\Omega(m)}{2^{m}} = \frac{\Omega(n^2)}{m 2^{m}}
\end{equation}
Letting $m=\log(n^2/\epsilon) = O(\log(n/\epsilon))$ yields
\begin{equation}
    \frac{1}{2^n} ||\oqft - \qft'||_F^2 \leq \epsilon
\end{equation}
which is what we desired to show.
\end{proof}

\begin{figure}
    \begin{center}
        \includegraphics{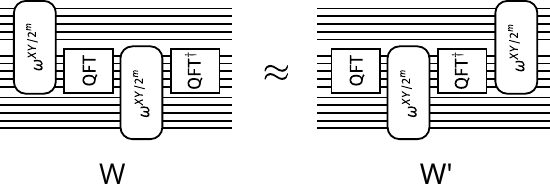}
    \end{center}
	\caption{The low-level ``optimistic commutation'' used to convert a linear-depth approximate QFT into the optimistic QFT (also presented in Figure~\ref{fig:oqft-detail}(b)).
    This transformation is performed on every even-indexed block of $m$ qubits; by Lemma~\ref{lem:commutation-error}, the error (measured in Frobenius norm) of making this replacement is small.
    Here, $\omega = \exp{(2\pi i /2^m)}$ and $X$ and $Y$ are the integer values of the computational basis states of the two blocks to which the gate is applied.}
	\label{fig:commutation-lemma}
\end{figure}

\begin{lemma}[Optimistic commutation of phase gates]
    For the operations $W$ and $W'$ depicted in Figure~\ref{fig:commutation-lemma}, applied $3m$ qubits, we have
    \begin{equation}
        ||W-W'||^2_F \leq 2^{2m} \cdot \Omega(m)
    \end{equation}
	\label{lem:commutation-error}
\end{lemma}

\begin{proof}
First, as per Definitions~\ref{def:opt-circuits} and~\ref{def:opt-circuits-basis-independent}, we have
\begin{equation}
    ||W - W'||_F^2 = \sum_{Y_{i-1}, X_i, X_{i+1}}^{2^m-1} |(W-W') (\ket{Y_{i-1}} \otimes \ket{X_i} \otimes \ket{X_{i+1}}) |^2
\end{equation}
where $Y_{i-1}, X_i, X_{i+1}$ are each $m$-bit integers.
(We denote the first (top) block of qubits in Figure~\ref{fig:commutation-lemma} as block $i-1$, the middle as $i$, and the last (bottom) as $i+1$ in reference to the larger context of Theorem~\ref{thm:oqft}).
Noting that both $W$ and $W'$ only apply phase rotations proportional to the values of $Y_{i-1}$ and $X_{i+1}$, we can rewrite the above sum in terms of new unitaries $\overline{W}_{Y_{i-1}, X_{i+1}}$ and $\overline{W}'_{Y_{i-1}, X_{i+1}}$ which only operate on block $i$, and the values of $Y_{i-1}$ and $X_{i+1}$ parametrize the unitaries:
\begin{equation}
    \label{eq:W-frobenius}
    ||W - W'||_F^2 = \sum_{X_i}^{2^m-1} \sum_{Y_{i-1}, X_{i+1}}^{2^m-1} |(\overline{W}_{Y_{i-1}, X_{i+1}}-\overline{W}'_{Y_{i-1}, X_{i+1}}) \ket{X_i}|^2
\end{equation}
Thus we now bound $|(\overline{W}_{Y_{i-1}, X_{i+1}}-\overline{W}'_{Y_{i-1}, X_{i+1}}) \ket{X_i}|^2$ as a function of $X_i$.

By inspection of Figure~\ref{fig:commutation-lemma}, we have
\begin{equation}
    \overline{W}_{Y_{i-1}, X_{i+1}} \ket{X_i} = \omega^{X_i Y_{i-1} / 2^m} \ket{\widetilde{X_i}}
\end{equation}
where $\ket{\widetilde{X_i}} = \sum \alpha_{X_i'} \ket{X_i'}$ is the result of applying $\qft$, then the second phase rotation $\omega^{X_{i+1} Y_i / 2^m}$, then $\qft^\dagger$ to $\ket{X_i}$ (see Figure~\ref{fig:oqft-detail}(b)).
This notation provides a straightforward way of expressing the action of $\overline{W}'_{Y_{i-1}, X_{i+1}}$ on $\ket{X_i}$:
\begin{equation}
\overline{W}'_{Y_{i-1}, X_{i+1}} \ket{X_i} = \omega^{X_i Y_{i-1} / 2^m} \sum_{X_i'} (\omega^{Y_{i-1}/2^{m}})^{(X_i'-X_i)} \alpha_{X_i'} \ket{X_i'}.
\end{equation}
We may thus write:
\begin{equation}
\label{eq:distance-sum}
|(\overline{W}_{Y_{i-1}, X_{i+1}}-\overline{W}'_{Y_{i-1}, X_{i+1}}) \ket{X_i}|^2 = \sum_{X_i'}^{2^{m}-1} |(1 - (\omega^{Y_{i-1}/2^{m}})^{(X_i'-X_i)})|^2 |\alpha_{X_i'}|^2.
\end{equation}
We proceed by separating this sum into two parts: the ``close'' part in which $|X_i'-X_i| < X_i^* = \min{(X_i, 2^m-X_i)}$ and thus $X_i'$ has not wrapped around $\bmod 2^m$; and the ``far'' part in which it has wrapped around.
The study of quantum phase estimation allows us to bound the size of both parts of this sum.

For the close part, when $|\alpha_{X_i'}|^2$ is large, we are guaranteed that $|X_i' - X_i|$ is (relatively) small because, by design, $0 < X_i'-X_i < 2^m-1$.
This implies that the phase difference $|(1 - (\omega^{Y_{i-1}/2^{m}})^{(X_i'-X_i)})|^2$ is small.
We now make this idea quantitative.
Using the fact that $|1-e^{i\delta}| \leq |\delta|$, we have
\begin{equation}
    |(1 - (\omega^{Y_{i-1}/2^{m}})^{(X_i'-X_i)})|^2 \leq |2\pi Y_{i-1} \Delta/2^{2m}|^2 = \frac{4\pi^2 Y_{i-1}^2}{2^{4m}} \Delta^2
\end{equation}
where $\Delta = X_i'-X_i$.
Meanwhile, from Equation~\ref{eq:block-approx} and quantum phase estimation (see Eq. 5.29 of Sec. 5.2.1 of \cite{nielsen_quantum_2011}), we have that
\begin{equation}
\label{eq:alpha-bound}
    |\alpha_{\Delta+X_i}| \leq \min{\left(\frac{1}{2|\Delta - X_{i+1}/2^m|}, 1 \right)}
\end{equation}
This implies that for the ``close'' part of the sum, we have:
\begin{equation}
\sum_{\substack{X_i'\\ |X_i'-X_i|<X_i^*}} |(1 - (\omega^{Y_{i-1}/2^{m}})^{(X_i'-X_i)})|^2 |\alpha_{X_i'}|^2 < \frac{4\pi^2 Y_{i-1}^2}{2^{4m+2}} \left[\sum_{\Delta=-X_i^*}^{-1} \frac{\Delta^2}{\Delta^2} + 2 + \sum_{\Delta=2}^{X_i^*} \frac{\Delta^2}{(\Delta - 1)^2}\right]
\end{equation}
Noting the following three facts:
\begin{equation}
\sum_{\Delta=2}^{X_i^*} \frac{\Delta^2}{(\Delta - 1)^2} = \sum_{\Delta=1}^{X_i^*-1} \frac{(\Delta+1)^2}{\Delta^2}
\end{equation}
\begin{equation}
    \sum_{\Delta=1}^\infty 1/\Delta^2 < 2
\end{equation}
\begin{equation}
    \label{eq:harmonic}
    \sum_{\Delta=1}^{X_i^*} 1/\Delta < \ln X_i^* + 1 < O(m)
\end{equation}

we have
\begin{equation}
\label{eq:close-sum}
\sum_{\substack{X_i'\\ |X_i'-X_i|<X_i^*}} |(1 - (\omega^{Y_{i-1}/2^{m}})^{(X_i'-X_i)})|^2 |\alpha_{X_i'}|^2 < \frac{4\pi^2 Y_{i-1}^2}{2^{4m+2}} \left(2X_i^* + O(n) \right)
\end{equation}

The ``far'' part of the sum is simpler to evaluate, because we only have a very loose bound on $|X_i' - X_i| < 2^m$ and thus $|(1 - (\omega^{Y_{i-1}/2^{m}})^{(X_i'-X_i)})| \leq 2 \pi Y_{i-1}/2^m$.
This is balanced by the fact that the $|\alpha_{X_i'}|$ are generally small.
We have (see~\cite{nielsen_quantum_2011} Section 5.2.1):
\begin{equation}
    \sum_{\Delta > X_i^*} |\alpha_{\Delta+X_i}|^2 < \frac{1}{2(X_i^*-1)}
\end{equation}
This yields
\begin{equation}
\label{eq:far-sum}
\sum_{\substack{X_i'\\ |X_i'-X_i|\geq X_i^*}} |(1 - (\omega^{Y_{i-1}/2^{m}})^{(X_i'-X_i)})|^2 |\alpha_{X_i'}|^2 < \frac{4\pi^2 Y_{i-1}^2}{2^{2m}} \frac{1}{2(X_i^*-1)}
\end{equation}

Now combining Eqs. \ref{eq:distance-sum}, \ref{eq:close-sum}, and \ref{eq:far-sum}, we have
\begin{equation}
|(\overline{W}_{Y_{i-1}, X_{i+1}}-\overline{W}'_{Y_{i-1}, X_{i+1}}) \ket{X_i}|^2 < \frac{2\pi^2 Y_{i-1}^2}{2^{2m}} \left(\frac{X_i^*}{2^{2m}} + \frac{1}{X_i^*-1} \right)
\end{equation}

Finally using this expression in Equation~\ref{eq:W-frobenius}, and again using Equation~\ref{eq:harmonic}, we have
\begin{equation}
    ||W - W'||_F^2 = 2\pi^2 \cdot 2^m \cdot \sum_{Y_{i-1}}^{2^m-1} \frac{Y_{i-1}^2}{2^{2m}} \sum_{X_i}^{2^m-1} \left(\frac{X_i^*}{2^{2m}} + \frac{1}{X_i^*-1} \right) < 2^{2m} \cdot \Omega(m)
\end{equation}
\end{proof}

\section{Alternative optimistic QFT construction}
\label{app:oqft-qft-only}

Here we show that the alternative construction of the optimistic QFT (see Algorithm~\ref{alg:oqft-qft-only} and Figure~\ref{fig:oqft-qft-only}), which is constructed entirely of smaller QFTs, is equivalent to the circuit presented in Algorithm~\ref{alg:oqft} and Figure~\ref{fig:oqft-detail}(a, ii).
The key idea is that a $2m$-qubit QFT is equivalent to a QFT on the first block of $m$ qubits, a phase rotation between the two blocks, and then an $m$-qubit QFT on the second block.
The transformation is depicted in Figure~\ref{fig:oqft-qft-only-proof}.

\begin{figure}[h]
	\begin{center}
		\includegraphics[width=\textwidth]{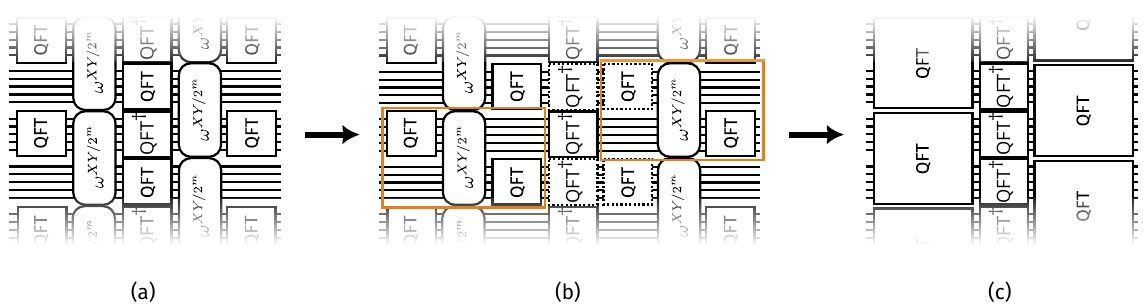}
	\end{center}
	\caption{Equivalence of the alternative optimistic QFT circuit (Alg.~\ref{alg:oqft-qft-only}) with the original circuit (Alg.~\ref{alg:oqft}). \textbf{(a)} The original optimistic QFT circuit. \textbf{(b)} Inserting the identity $\qft \cdot \qft^\dagger$ (dotted borders), the phase rotation blocks can be grouped with neighboring $m$-qubit QFT blocks (orange boxes) to form $2m$-qubit QFT blocks. \textbf{(c)} The alternative form of the optimistic QFT circuit.}
	\label{fig:oqft-qft-only-proof}
\end{figure}

\section{Proof of factoring via optimistic QFT}
\label{app:shor-oqft-proof}

Here we show that using optimistic QFTs to build an optimistic multiplier, as described in Section~\ref{sec:opt-qft-factoring}, is sufficient for factoring with high probability, even without using the reductions described in Section~\ref{sec:oqft-reduction}.

\begin{restatable}[Optimistic multiplier]{theorem}{opt-multiplier}
    \label{thm:opt-mult}
	Consider the circuit described in~\cite{kahanamoku-meyer_fast_2024}, Appendix E; which implements the unitary $U \ket{y} \to \ket{cy \bmod N}$ for $y<N$, and $c$ and $N$ classical integers having $\gcd(c, N) = \gcd(c-1, N) = 1$.
    Replacing all QFTs in that circuit with optimistic QFTs of error $\epsilon$ yields an optimistic multiplier with error parameter $O(\epsilon)$.
\end{restatable}

\begin{proof}
    Inside of the multiplier, QFTs are applied to computational basis states $\ket{y}$, $\ket{cy \bmod N}$, and $\ket{w}$ where $w \approx ((c-1)y \bmod N)/N$ expressed as an $n + O(\log n)$ bit binary fraction.
	Since $c$ and $c-1$ are coprime to $N$, both $\ket{cy \bmod N}$ and the first $n$ bits of $\ket{w}$ will respectively take on a unique value for each input $y<N$.%
    \footnote{The last $O(\log n)$ bits of $w$ can be handled exactly by appending a small exact QFT circuit plus some phase rotations, so they can be ignored for the error analysis.}
	Thus when the multiplier is applied to each of the $N-1$ different computational basis states $\ket{y}$ having $0 < y < N$, each $n$-qubit QFT inside the multiplier will explore a new vector in its input Hilbert space.
    Since $N/2^n = O(1)$, by Lemma~\ref{lem:large-error-subspaces} the Frobenius error of the optimistic QFT on a subspace of dimension $N$ is $O(\epsilon)$.
    Thus by triangle inequality we may bound the Frobenius error of the optimistic multiplier as the sum of the Frobenius errors of the optimistic QFTs, proving the theorem.
\end{proof}

\begin{restatable}[Factoring with the optimistic multiplier]{theorem}{shoroqftthm}
	Consider Shor's algorithm for integer factorization~\cite{shor_polynomial-time_1997}, which applies quantum period finding to the modular exponentiation unitary $U \ket{x} \ket{r} \to \ket{x} \ket{rg^x \bmod N}$ for classical integers $g$ and $N$.
    If that unitary is implemented via $O(n)$ applications of the optimistic multiplier of Theorem~\ref{thm:opt-mult} (as described in~\cite{kahanamoku-meyer_fast_2024}, Appendix E), and the modular exponentiation circuit is applied to the initial state $\sum_x \ket{x} \ket{r}$ for some classically-chosen random integer $r<N$, the probability of successfully finding the factors is at least $p \geq O(1) \cdot (1-O(n\sqrt{\epsilon}))$.
\end{restatable}

\begin{proof}
	We may express any one optimistic multiplication $\widetilde{\mathcal{U}}_\times$ as
    \begin{equation}
        \widetilde{\mathcal{U}}_\times \ket{y} = \ket{cy \bmod N} - \varepsilon_y \ket{e_y}
    \end{equation}
	where $\varepsilon_y \ket{e_y} = - (\widetilde{\mathcal{U}}_\times-\mathcal{U}_\times)\ket{y}$ is the error vector.
    Note that $\ket{e_y}$ is not orthogonal to $\ket{cy \bmod N}$, so the overall state still has norm 1.
    For clarity of exposition we have pulled out a factor $\varepsilon_y$ so that $\ket{e_y}$ also is normalized; by absorbing any phase into $\ket{e_y}$ we may take $\varepsilon_y$ to be a nonnegative real number.
    It is important to note here that for worst-case $y$, $\varepsilon_y$ may not be well-bounded---later, we will find that the \textit{average} value of $\varepsilon_y$, which we can bound well, is the important quantity.

    The overall modular exponentiation circuit can be expressed similarly as
    \begin{equation}
        \ket{x}\ket{r} \to \ket{x}(\ket{rg^x \bmod N} - \varepsilon_{x,r}' \ket{e_{x,r}'})
    \end{equation}
    for some error vector $\varepsilon_{x,r}' \ket{e_{x,r}'}$.
    (Importantly, the bits of $x$ are only used as controls for phase rotations; they are not touched by the optimistic QFTs making up the multiplier circuit so the error is only on the output register.)
    Note that by the triangle inequality, $\varepsilon_{x,r}' \leq \sum_{y_i} \varepsilon_{y_i}$ where the sum is over intermediate values $y_i$ encountered during the modular exponentiation.

    Applying this modular exponentiation to the initial state as described in the theorem yields
	\begin{equation}
		\ket{\psi_\mathsf{opt}} \equiv \widetilde{\mathcal{U}}_\mathsf{opt} \left[\frac{1}{2^{n/2}} \sum_x \ket{x} \ket{r} \right] = \frac{1}{2^{n/2}} \sum_{x} \ket{x} (\ket{rg^x \bmod N} - \varepsilon_{x,r}' \ket{e_{x,r}'})
	\end{equation}
    The ``target'' state which would result from an exact modular exponentiation is
    \begin{equation}
        \ket{\psi_\mathsf{Shor}} = \frac{1}{2^{n/2}} \sum_x \ket{x} \ket{rg^x \bmod N}.
    \end{equation}
    Since applying period finding to $\ket{\psi_\mathsf{Shor}}$ yields the factors of $N$ with probability $O(1)$, the probability $p$ of successfully finding the factors using $\ket{\psi_\mathsf{opt}}$ is proportional to
    \begin{equation}
        p \propto \mathbb{E}_r \left[ |\braket{\psi_\mathsf{Shor}|\psi_\mathsf{opt}}|^2 \right] \geq \mathbb{E}_r \left[ \left(\frac{1}{2^n} \sum_x (1 -\varepsilon_{x,r}')\right)^2 \right] \geq 1 - \frac{2}{2^n} \sum_x \mathbb{E}_r \left[ \varepsilon_{x,r}' \right].
    \end{equation}

    Now recall from above that $\varepsilon_{x,r}' \leq \sum_{y_i} \varepsilon_{y_i}$, where $y_i$ are the intermediate values encountered during the modular exponentiation for a particular $x$ and $r$.
    Then $\mathbb{E}_r \left[ \varepsilon_{x,r}' \right] \leq \sum_{y_i} \mathbb{E}_r \left[ \varepsilon_{y_i} \right]$.
    Since for all $i$, $y_i = rz$ for some integer $z$ which does not depend on $r$, the $y_i$ are (independently) random elements of the multiplicative group modulo $N$, and thus by Theorem~\ref{thm:opt-mult} and Definition~\ref{def:opt-circuits}, $\mathbb{E}_r \left[ \varepsilon_{y_i} \right] \leq O(\sqrt{\epsilon})$.
    There are $O(n)$ multiplications in the modular exponentiation, so $\mathbb{E}_r \left[ \varepsilon_{x,r}' \right] \leq O(n) \cdot \mathbb{E}_r \left[ \varepsilon_{y_i} \right] \leq O(n\sqrt{\epsilon})$.
    Thus $\mathbb{E}_r \left[ |\braket{\psi_\mathsf{Shor}|\psi_\mathsf{opt}}|^2 \right] \geq 1-O(n\sqrt{\epsilon})$ and finally $p \propto 1 - O(n\sqrt{\epsilon})$, proving the theorem.
\end{proof}

\textbf{Remark---}
Intuitively, the random factor $r$ should not even be necessary, as we do not expect the intermediate values in the modular exponential to have any particular structure that biases them towards basis states with high error in the optimistic QFT.

\section{Proofs from Section \ref{sec:oqft-reduction}}
\label{sec:sec-5-pfs}
In this appendix we provide the proofs deferred from Section \ref{sec:oqft-reduction}, starting with the approximate controlled phase gradient.

\approxpg*{}
\begin{proof}
    The construction is extremely similar to the standard approximate QFT \cite{coppersmith_approximate_2002}. First, note that
    \[e^{2\pi i x z / 2^n} \ket{x}\ket{z} =
    \ket{x} \left ( \bigotimes_{j=0}^{n-1} e^{2\pi i (0.x_j \dots x_{n-1})\cdot z_j} \ket{z_j} \right).\]
    In a similar manner to the QFT, this can be implemented using $O(n^2)$ controlled $Z$ rotation gates, many of which will be rotations through very small angles. Just as in the approximate QFT, we simply drop the small rotation gates, and implement the transformation
    \[\ket{x}\ket{z} \mapsto \ket{x} \left ( \bigotimes_{j=0}^{n-1} e^{2\pi i (0.x_j \dots x_{j+m-1})\cdot z_j} \ket{z_j}\right )\]
    where $m = \log(2 \pi n / \delta)$. This can be performed with $O(n\log(n/\delta)$ controlled $Z$ rotation gates. Furthermore, by grouping these gates in the correct way, the circuit can be implemented in depth $O(\log(n/\delta))$. In the first time step, we perform the required two-qubit gate between qubit $j$ of the first register and qubit $j$ of the second register, for $0\le j \le n-1$. In the second time step, we perform the required two-qubit gate between qubit $j$ of the first register and qubit $j-1$ of the second register, for $1 \le j \le n-1$, and so on, for $m$ steps.

    To bound the error, we note that we must simply sum the operator norms between each of the dropped controlled rotations and the identity. We see that for a controlled $Z$ rotation of angle $\alpha$, $\| I - \text{diag}(1, 1, 1, e^{i\alpha})\| = |1 - e^{i\alpha}| \le |\alpha|$.

    Thus, the total error is bounded by
    \begin{equation*}
        \sum_{j=0}^{n-m-1} \sum_{k=m+j}^{n-1}\pi 2^{-(k-j)} \le \pi \sum_{j=0}^{n-m-1} 2^{-m+1}
        \le 2\pi n 2^{-m}
        = \delta.
    \end{equation*}
\end{proof}

Next, we give further details on the construction of Theorem \ref{thm:purified-approx-qft}.

\paqft*{}
\begin{proof}
    We must now implement $V'(r_1, r_2) = \sum_{r_1, r_2} \ket{r_1, r_2}\bra{r_1, r_2} \otimes X_{2^n}^{r_1} Z_{2^n}^{r_2}$, which corresponds to controlled versions of the Weyl-Heisenberg operators. The operator $\sum_{r_1} \ket{r_1}\bra{r_1} \otimes X_{2^n}^{r_1}$ is simply a quantum-quantum addition,
    which we implement using the construction from~\cite{takahashi_quantum_2009}. This circuit has depth $O(\log n)$ and requires $O(n/\log n)$ ancillas. The operator $\sum_{r_2} \ket{r_2}\bra{r_2} \otimes Z_{2^n}^{r_2}$ is the controlled phase gradient operation discussed in Lemma \ref{lem:approxpg}.

    Let $\oU$ denote the derandomized optimistic QFT circuit with perfect controlled phase gradients and error parameter $\epsilon/9$, and let $\U'$ denote an implementation of $\oU$ using the approximate controlled-phase gradient operations as described in Lemma \ref{lem:approxpg} with error parameter $\sqrt{\epsilon}/3$. By the triangle inequality and Theorem \ref{thm:purification}, for all states $\ket{\psi}$
    \[\left| (\U' - \mathbb{I} \otimes \U) \left(\frac{1}{\sqrt{k}} \sum_i \ket{i} \otimes \ket{\psi}\right) \right|^2 \leq \epsilon.\]

    Thus, $\U'$ satisfies the conditions of the theorem.
\end{proof}

\end{document}